\documentclass{article}

\usepackage[final,nonatbib]{neurips_2022}

\usepackage{multibib}
\newcites{SM}{Supplementary References}
\usepackage{cite}
\usepackage[utf8]{inputenc} 
\usepackage[T1]{fontenc}    
\usepackage{hyperref}       
\usepackage{url}            
\usepackage{booktabs}       
\usepackage{amsfonts,amsmath,amssymb}       
\usepackage{nicefrac}       
\usepackage{microtype}      
\hypersetup{colorlinks=true,linkcolor=[RGB]{0,80,255},citecolor=red!80!blue}
\usepackage[cmyk]{xcolor}
\usepackage{graphicx}
\usepackage{wrapfig}
\usepackage{bm}
\usepackage{float}
\usepackage{bbm}
\usepackage[font=small,labelfont=bf]{caption}

\usepackage{subfig,amsthm,mathtools}
\usepackage{algorithm}
\usepackage{algpseudocode}

\newtheorem{thm}{Theorem}
\newtheorem{cor}{Corollary}
\newtheorem{prop}{Proposition}

\DeclareMathOperator{\tr}{Tr}

\newcommand{\argmin}[1]{\underset{#1}{\operatorname{arg}\,\operatorname{min}}\;}

\newcommand{\e}{{\epsilon}}

\newcommand{\n}{{\bf n}}

\newcommand{\uvec}{{\bf u}}
\newcommand{\vvec}{{\bf v}}
\newcommand{\w}{{\bf w}}
\newcommand{\x}{{\bf x}}
\newcommand{\y}{{\bf y}}
\newcommand{\z}{{\bf z}}

\newcommand{\A}{{\bf A}}
\newcommand{\B}{{\bf B}}
\newcommand{\C}{{\bf C}}

\newcommand{\F}{{\bf F}}

\newcommand{\I}{{\bf I}}

\newcommand{\M}{{\bf M}}

\newcommand{\Q}{{\bf Q}}
\newcommand{\R}{\mathbb{R}}

\newcommand{\U}{{\bf U}}
\newcommand{\V}{{\bf V}}
\newcommand{\W}{{\bf W}}
\newcommand{\X}{{\bf X}}
\newcommand{\Y}{{\bf Y}}
\newcommand{\Z}{{\bf Z}}

\newcommand{\bSig}{\boldsymbol{\Sigma}}

\newcommand{\bmu}{\boldsymbol{\mu}}

\renewcommand{\v}{{\bf v}}

\newcommand{\bXi}{\boldsymbol{\Xi}}

\title{Constrained Predictive Coding as a Biologically Plausible Model of the Cortical Hierarchy}

\author{%
    Siavash Golkar$\thanks{Equal contribution}\hspace{3pt}^{\,1}$\\
    \And Tiberiu Tesileanu$\protect\footnotemark[1]\hspace{3pt}^{\,1}$\\
    \And Yanis Bahroun$^{\,1}$\\
    \And Anirvan M.~Sengupta$^{\,2,3,4}$\\
    \And Dmitri Chklovskii$^{\,1,5}$\vspace{14pt}\\
$^{1\,}$Center for Computational Neuroscience, Flatiron Institute
\\
$^{2\,}$Center for Computational Mathematics, Flatiron Institute
\\
$^{3\,}$Center for Computational Quantum Physics, Flatiron Institute
\\
$^{4\,}$Department of Physics and Astronomy, Rutgers University
\\
$^{5\,}$Neuroscience Institute, NYU Medical Center
\vspace{5pt}\\
\texttt{\{sgolkar,ttesileanu,ybahroun,dchklovskii\}@flatironinstitute.org}\\
\texttt{anirvans.physics@gmail.com }}

\begin{document}

\maketitle

\begin{abstract}
Predictive coding (PC) has emerged as an influential normative model of neural computation, with numerous extensions and applications. As such, much effort has been put into mapping PC faithfully onto the cortex, but there are issues that remain unresolved or controversial. In particular, current implementations often involve separate value and error neurons and require symmetric forward and backward weights across different brain regions. These features have not been experimentally confirmed. In this work, we show that the PC framework in the linear regime can be modified to map faithfully onto the cortical hierarchy in a manner compatible with empirical observations. By employing a disentangling-inspired constraint on hidden-layer neural activities, we derive an upper bound for the PC objective. Optimization of this upper bound leads to an algorithm that shows the same performance as the original objective and maps onto a biologically plausible network. The units of this network can be interpreted as multi-compartmental neurons with non-Hebbian learning rules, with a remarkable resemblance to recent experimental findings. There exist prior models which also capture these features, but they are phenomenological, while our work is a normative derivation. Notably, the network we derive does not involve one-to-one connectivity or signal multiplexing, which the phenomenological models required, indicating that these features are not necessary for learning in the cortex. The normative nature of our algorithm in the simplified linear case allows us to prove interesting properties of the framework and analytically understand the computational role of our network's components. The parameters of our network have natural interpretations as physiological quantities in a multi-compartmental model of pyramidal neurons, providing a concrete link between PC and experimental measurements carried out in the cortex.
\end{abstract}


\section{Introduction}\label{sec:intro}

Over the past decades, predictive coding (PC), a normative framework for learning representations that maximize predictive power, has played an important role in computational neuroscience~\cite{Rao1999,Friston2005}. Initially proposed as an unsupervised learning paradigm in the retina~\cite{Srinivasan}, it has since been expanded to the supervised regime~\cite{Whittington2017} with arbitrary graph topologies~\cite{PConarbitrarygraphs,DGPC}. 
The PC framework has been analyzed in many contexts~\cite{predictivebalance1,predictivebalance2} and has found many applications, from clinical neuroscience~\cite{clinicalPCN, HabNRef} to memory storage and retrieval~\cite{AssociativePC}.  We refer the reader to \cite{PCreview,Millidgereview} for recent reviews.

\begin{figure}[!tbh]
    \center\includegraphics{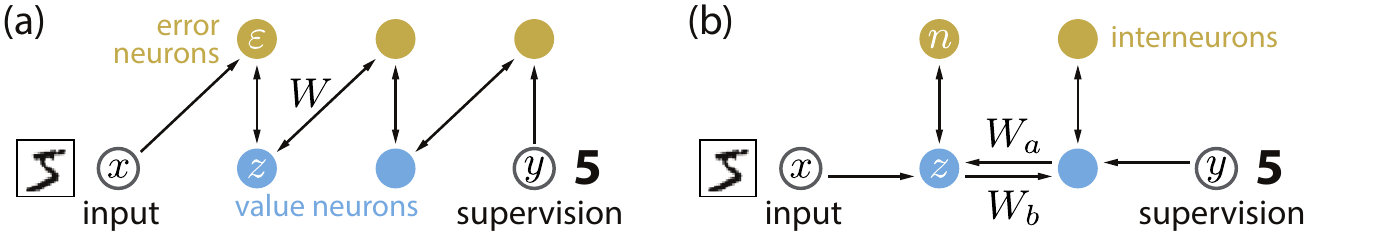}
    \caption{Schematic architecture of the predictive coding network (PCN) and our covariance-constrained network (BioCCPC). (a) PCN from~\cite{Whittington2017}, figure adapted from~\cite{Millidgereview}. The intra-layer connectivity is to be one-to-one, while the inter-layer connectivity is symmetric. (b) Our BioCCPC network. There is no requirement for symmetric weights across layers, and the connectivity within layers can be diffuse. \label{fig:schematic}}
\end{figure}
Predictive coding is viewed as a possible theory of cortical computation, and many parallels have been drawn with the known neurophysiology of cortex~\cite{PCreview}. While the initial works proposed a biologically plausible network~\cite{Rao1999,Friston2005}, the connection with cortex was more closely examined in~\cite{Bastos2012}, where the PC module was mapped onto a cortical-column microcircuit. However, there are aspects of this mapping that have proved controversial. Among these are the requirement of multiple redundant cortical operations, the symmetric connectivity pattern, the one-to-one connectivity of value and error neurons, and also the requirement that feedback connections be inhibitory~\cite{PCreview,Heilbron,Kogo}, as sketched in Fig.~\ref{fig:schematic}a . The presence of separate error and value neurons has itself been called into question~\cite{Heilbron}.

The PC-based neural circuits also do not account for more recent experimental findings highlighting the details of computation in the cortex~\cite{Spruston2008,Larkum2013,sjostrom2006,golding2002,gambino2014sensory,pyramidal_review,bittner2015,bittner2017behavioral,MageeGrienberger2020,hardie2009synaptic}. For example, the learning dynamics of multi-compartmental pyramidal neurons has been closely investigated. In particular, it was observed that the plasticity of the synapses of the basal compartment is driven by the activity in the apical tuft of the neuron by so-called calcium plateau potentials~\cite{golding2002,sjostrom2006,bittner2017behavioral}, leading to non-Hebbian learning rules~\cite{MageeGrienberger2020}.
%
%
%
These experiments have motivated the development of several models of microcircuits with multi-compartmental neurons~\cite{kording2001supervised,Urbanczik2014,guergiuev2016deep,Sacramento2018,Payeur2020,Gidon2020,Milstein2020}. In a number of cases, it has been shown that these models can replicate learning similar to the backpropagation algorithm under specific assumptions~\cite{backprop,goodfellow2016deep,Sacramento2018}. However, because of their rather phenomenological nature, detailed analysis of these models is in many cases challenging and one must resort to numerical simulations, rendering the task of understanding the role of various neurophysiological quantities difficult. To this date, a normative framework (PC or otherwise) that can explain these experimental findings is still lacking.

In this work, we show how the PC framework can be made compatible with the aforementioned experimental observations. Inspired by prior work which explored the effects of finding decorrelated representations~\cite{choi2005blind,pehlevan2019neuroscience,wanner2020whitening,DesjardinsNG,Fujimoto_2018,Generalized_Whitening_Neural_Net}, we add a decorrelating inequality constraint to the covariance of the latent representations of PC. Using this constraint, we derive an upper bound for the PC objective.  By working in the linear regime, we can prove interesting properties of our algorithm. We show that the learning algorithm derived from this upper bound does not suffer from the issues of prior implementations and naturally maps onto the hierarchical structure of the cortical pyramidal neurons.

\paragraph{Contributions}
\begin{itemize}\vspace{-2mm}
\setlength\itemsep{-0.05em}
    \item By imposing a decorrelation-inspired inequality constraint on the latent space, we find an analytic upper bound to the PC objective.
    \item We introduce BioCCPC, a novel biologically plausible algorithm derived from this upper bound and show that it closely matches the known physiology of the cortical hierarchy.
    \item We interpret the different parameters of the algorithm in terms of the conductances and leaks of the separate compartments of the pyramidal neuron. We find that the neural compartmental conductances encode the varainces of the PC framework, and the somatic leak maps onto a thresholding mechanism of the associated eigenvalue problem.
\end{itemize}

\section{Related work and review of predictive coding}


The backpropagation algorithm~\cite{Werbos:74, backprop} is the predominant tool in the training of deep neural networks but is generally not considered biologically plausible~\cite{crick1989recent}. Over the years, many authors have explored biologically plausible approximations or alternatives to backpropagation~\cite{LeCun1986,Lecun1987,Kolen1994,Lee2014,Lillicrap2016,Scellier2017, Sacramento2018, Krotov2019,Akrout2019,Guerguiev2019,Payeur2020, Ernoult2020,signsymmetry,RRR,RRR_workshop,differentialTP,DendriticComputing} (for a more complete review see~\cite{lillicrap2020backpropreview}). These approaches generally fall into two categories. First are normative approaches, such as Predictive Coding~\cite{Rao1999,Friston2005,Srinivasan}, Target Propagation and variations~\cite{LeCun1986,Lecun1987,Lee2014}, Equilibrium Propagation~\cite{Scellier2017, Ernoult2020} and others~\cite{signsymmetry,Akrout2019,RRR,RRR_workshop}, where one starts from a mathematically motivated optimization objective. These methods, by virtue of their normative derivation, have a firm grounding in theory; however, they do not fully conform to the experimental observations in the brain (see below and Sec.~\ref{sec:bio}). 
%
The second approach is driven by biology, with network structures and learning rules inspired by experimental findings~\cite{Urbanczik2014, Lillicrap2016,Sacramento2017, Sacramento2018,Payeur2020}. While these works mostly conform to the experimentally observed findings, they are more challenging to analyze because of their conjectured phenomenological nature.

The goal of the present paper is \emph{not} to propose yet another biologically plausible alternative to backpropagation. It is rather to demonstrate that the normative framework of predictive coding, when combined with a constraint, can indeed closely match experimental observations. For this reason, in this work, we focus on comparing our method with previous implementations of predictive coding and do not concern ourselves with other biologically plausible alternatives to backpropagation. The relationship between the PC framework and backpropagation was explored in~\cite{Whittington2017,PCarbitraryPoor,PCexactBackprop,salvatori-PCbackprop,rosenbaum2021}. The advantages of PC over backpropagation were highlighted in~\cite{Millidgereview}.

\paragraph{Notation.}Bold upper case $\M$ and lower case variables $\v$ denote matrices and vectors, respectively. By upper case letter $\X,\Y,\Z$, we denote data matrices in $\R^{d\times T}$, where $d$ and $T$ are the dimensions of the relevant variable and the number of samples. Lower case $\x,\y,\z$, denote the relevant quantities of a single sample, and $\x_t,\y_t,\z_t$, denote the $t^\text{th}$ sample. $\|\M \|^2_F$ denotes the Frobenius norm of \M.

\subsection{Review of predictive coding}\label{sec:PC}


\paragraph{Probabilistic model.} 
In this section, we review the supervised predictive coding algorithm~\cite{Whittington2017}. The derivation starts from a probabilistic model for supervised learning, which parallels the architecture of an artificial neural network (ANN) with $n+1$ layers. In this model, the neurons of each layer are random variables $\z^{(l)}$ (denoting the vector of activations in the $l^\text{th}$ layer) with layers $0$ and $n$, respectively, denoting the input and output layers of the network.
We assume that the joint probability of the latent variables factorizes in a Markovian structure
\begin{equation*}
    p(\z^{(0)}, \z^{(1)}, \cdots , \z^{(n)})= p(\z^{(n)}|\z^{(n-1)})\times p(\z^{(n-1)}|\z^{(n-2)})\times \cdots \times p(\z^{(0)})
\end{equation*}
with the relationship between the random variables of adjacent layers given by:
\begin{equation}
    p\bigl(\z^{(l)}\vert\z^{(l-1)}\bigr)=\mathcal N\bigl(\z^{(l)};\bm \mu^{(l)},\bSig^{(l)}\bigr)\,, \qquad \text {with}\quad \bmu^{(l)} =  \W^{(l-1)} f(\z^{(l-1)}), \quad \bSig^{(l)} = \sigma^{(l)2} \I\,.
\end{equation}
The mean of the probability density on layer $l$ mirrors the activity of the analogous ANN given by
    $\bmu^{(l)} =  \W^{(l-1)} f(\z^{(l-1)}),$
where 
$\W^{(l-1)}$ are the weights connecting layers $l-1$ and $l$.
The objective function is then given by the negative log-likelihood of the joint distribution function:
\begin{align}\label{eq:WB_obj}
    L = -\sum_t \log p(\z_t^{(0)},\dotsc,\z_t^{(n)}) &= \frac12 \sum_{l} \frac{\bigl\| \Z^{(l)}-\W^{(l-1)} f(\Z^{(l-1)})\bigr\|_F^2}{\sigma^{(l)2}} + \text{const}~,
\end{align}
where we have switched to the data matrix notation for brevity and assumed the variances $\sigma^{(l)2}$ are fixed hyperparameters. In the following, we refer to $L$, Eq.~\eqref{eq:WB_obj},  without the constant term.


\paragraph{Learning.} Learning takes place in two steps. First, the values of the random variables are determined by finding the most probable configuration of the joint distribution function when both the input and output layers are conditioned on the given input and output ($\z^{(0)} = \x$ and $\z^{(n)} = \y$):
\begin{equation}\label{eq:Step_1_PC}
    \z^{*(1)},\dotsc,\z^{*(n-1)} = \argmin{\z^{(1)},\dotsc,\z^{(n-1)}}  L(\z^{(0)}= \x,\z^{(n)} = \y).
\end{equation}
The solution to this minimization problem can be found via gradient descent, which we evaluate component-wise for clarity as: 
\begin{equation}
    \dot z^{(l)}_j = -\eta \partial_{z_j} L = \eta( -\varepsilon^{(l)}_j + \sum_i \varepsilon^{(l+1)}_i W_{ij}^{(l)}f'( z_j^{(l)}) ),\quad \text{with} \quad 
    \varepsilon^{(l)}_i = \frac{ z_i^{(l)}-\mu_i^{(l)}}{\sigma^{(l)2}}\,,
\end{equation}
where $\eta$ is the gradient descent step size. The second step is then to minimize the objective with respect to the weights while keeping the previously obtained neuron values fixed. This corresponds to optimizing the value of the loss at the MAP estimate and can also be carried out by gradient descent. 
This algorithm can be implemented by a biologically plausible network as described in~\cite{Whittington2017}; see Figure~\ref{fig:schematic}a. However, as discussed in Section~\ref{sec:intro}, its mapping onto the cortex has proved controversial. For further details regarding these steps see \cite{Whittington2017}.

\section{A constrained predictive coding framework}

In this section, we introduce and discuss our novel, covariance-constrained predictive coding (CCPC) model within the supervised PC paradigm of~\cite{Whittington2017}. Our model also straightforwardly extends to the unsupervised learning paradigm discussed in~\cite{Rao1999,Friston2005}. For simplicity we work in the linear regime~($f(x)=x$) which allows us to prove different properties of our framework.




\subsection{Derivation of upper bound objective}

\paragraph{Reduction to a sum of objectives.}
We start by reducing the optimization problem, Eq.~\eqref{eq:WB_obj}, into a set of overlapping sub-problems, which will allow us to break the symmetry between feedforward and feedback weights. 
To do so, we first introduce a copy of the terms containing the weights $\W^{(1)}$ to $\W^{(n-2)}$, denoted by $\W_a^{(l)}$ and $\W_b^{(l)}$ respectively, as 
\begin{equation}
    \label{eq:obj_doubled}
    \begin{split}
    \min_{\Z,\W} L &= \min_{\Z,\W} \frac 12\sum_{l=1}^{n} \frac{\bigl\|\Z^{(l)}-\W^{(l-1)} \Z^{(l-1)}\bigr\|_F^2}{\sigma^{(l)2}} = \min_{\Z,\W_a,\W_b}  \hat L\\
    \hat L &= \min_{\Z,\W_a,\W_b}  \frac 12\frac{\bigl\|\Z^{(1)}-\W_b^{(0)} \Z^{(0)}\bigr\|_F^2}{\sigma^{(1)2}} + 
    \frac 12\frac{\bigl\|\Z^{(n)}-\W_a^{(n-1)} \Z^{(n-1)}\bigr\|_F^2}{\sigma^{(n)2}}\\ 
    & \qquad + \frac14  \sum_{l=2}^{n-1} \left[ \frac{\bigl\|\Z^{(l)}-\W_b^{(l-1)} \Z^{(l-1)})\bigr\|_F^2}{\sigma^{(l)2}} + 
    \frac{\bigl\|\Z^{(l)}-\W_a^{(l-1)} \Z^{(l-1)})\bigr\|_F^2}{\sigma^{(l)2}}\right].
    \end{split}
\end{equation}
For consistency, we rename $\W^{(0)}$, $\W^{(n-1)}$ to $\W_b^{(0)}$, $\W_a^{(n-1)}$, respectively.
Introducing these copies does not change the optimization\footnote{This can be directly verified by finding the optima for $\W$'s before and after the change. We have $\W^{(l)} = \W_a^{(l)} = \W_b^{(l)} = \Z^{(l+1)}\Z^{(l)\top} (\Z^{(l+1)}\Z^{(l)\top})^{-1}$. Plugging these back into Eq.~\eqref{eq:obj_doubled} we see that the equality holds. However, in the next step, since we treat $\W_a$'s and $\W_b$'s differently, $\W_a = \W_B$ will no longer hold.} but will help us avoid weight sharing in the steps below. We now pair the terms two by two as 
\begin{align}
    \label{eq:deep_paired}
    \min_{\Z,\W_a,\W_b} \hat L = \frac 12\sum_{l=1}^{n-1} \left[g_b^{(l)} \bigl\|\Z^{(l)}-\W_b^{(l-1)} \Z^{(l-1)}\bigr\|_F^2 + g_a^{(l)}\bigl\|\Z^{(l+1)}-\W_a^{(l)} \Z^{(l)}\bigr\|_F^2\right]~,
\end{align}
with $g_a^{(n-1)} = 1 / \sigma^{(n)2}$, $g_b^{(1)} = 1 / \sigma^{(1)2}$, and $g_a^{(l-1)} = g_b^{(l)} = 1 / (2\sigma^{(l)2})$ for $l = 2, \dotsc, n-1$. 

Weight sharing occurs here from terms like $\z^{(l+1)\top} \W \z^{(l)}$, obtained from expanding the squared norms in Eq.~\eqref{eq:deep_paired}. Indeed, the gradient descent dynamics with respect to $\z^{(l+1)}$ (resp. $\z^{(l)}$) leads to terms of the form $\W \z^{(l)}$ in $\dot \z^{(l+1)}$ (resp. $\W^\top \z^{(l+1)}$ in $\dot \z^{(l)}$), which use the same weights $\W$. Thanks to the doubling of the weights, we can avoid this problem by optimizing each term in the sum in Eq.~\eqref{eq:deep_paired} separately. In other words,
\begin{gather}
    \label{eq:deep_split}
    \min_{\Z,\W_a,\W_b} \hat L  \leq \sum_{l=1}^{n-1} \min_{\Z^{(l)},\W_a^{(l)},\W_b^{(l-1)}} \hat L^{(l)}~,\\
    \text{where} ~~\hat L^{(l)}\equiv  \frac 12\sum_{l=0}^{n} \left[g_b^{(l)} \bigl\|\Z^{(l)}-\W_b^{(l-1)} \Z^{(l-1)}\bigr\|_F^2 + g_a^{(l)}\bigl\|\Z^{(l+1)}-\W_a^{(l)} \Z^{(l)}\bigr\|_F^2\right].\notag
\end{gather}
This equality holds simply because we are no longer finding the minimum of the full objective $\hat L$. We are instead finding the minimum of each component separately, and then evaluating $\sum_l \hat L^{(l)} = \hat L$. 
This splits the $(n+1)$-layer optimization problem into a set of $3$-layer optimizations, in each of which only the middle layer is being optimized. Note, however, that these are overlapping, so the different optimization problems need to be solved self-consistently. We make this precise in the supplementary materials section and show that it provides an upper bound for our objective $L$ (SM Sec.~\ref{app:upperbound}).
Separating the objective function in this manner eliminates the weight sharing problem for $\W_b$, but the problem remains for $\W_a$. We address this problem in the following.


\paragraph{Whitening constraint.} 

The idea of decorrelating internal representation has been widely used for unsupervised tasks, often motivated by neuroscience~\cite{choi2005blind,pehlevan2019neuroscience,wanner2020whitening}. 
In the case of deep learning, the main motivations were improved convergence speed and generalization~\cite{DesjardinsNG,Fujimoto_2018,Generalized_Whitening_Neural_Net}. 
Decorrelation has also been used to circumvent the weight transport problem~\cite{RRR}.
Inspired by these observations, we introduce the constraint  $\frac 1T \Z^{(l)} \Z^{(l)\top} \preceq \I$, imposing an upper bound on the eigenvalues of the covariance matrix. The inequality can be implemented by using a positive-definite Lagrange multiplier $\Q^{(l)\top}\Q^{(l)}$ (for details see SM Sec. \ref{app:addition}):
\begin{equation}
    \label{eq:linear_final}
    \begin{split}
    \min  \hat L^{(l)} &\leq 
    \min_{\Z^{(l)},\W_a^{(l)},\W_b^{(l-1)}}\max_{\Q^{(l)}} \frac 12\biggl[g_b^{(l)} \Bigl\|\Z^{(l)}-\W_b^{(l-1)}\Z^{(l-1)}\Bigr\|_F^2 
    + g_a^{(l)}  T \, \W_a^{(l)\top} \W_a^{(l)} \\
    &
    + g_a^{(l)}   \tr \Bigl(- 2 \Z^{(l+1)\top}\W_a^{(l)}\Z^{(l)}\Bigr)  
     + \tr \Q^{(l)\top}{\Q^{(l)}}\bigl({\Z^{(l)}}\Z^{(l)\top} - T\I\bigr)+ c^{(l)} \bigl\|\Z^{(l)}\bigr\|_F^2 
     \biggr]\,,
    \end{split}
\end{equation}
where we have added an additional quadratic term in $\Z$ as a regularizer. 


\subsection{Neural dynamics and learning rules.} Similar to the PCN of \cite{Whittington2017}, the dynamics of our network during learning proceeds in two steps. First the neural dynamics is derived by taking gradient steps of $\hat L^{(l)}$ from Eq.~\eqref{eq:linear_final} with respect to $\z^{(l)}$:
\begin{align}\label{eq:lin_neural_dynamics}
    \dot \z^{(l)} = g_b^{(l)} \W_b^{(l-1)}\z^{(l-1)}+g_a^{(l)} \W_a^{(l)\top}\z^{(l+1)} -  \bigl(g_b^{(l)}+c^{(l)}\bigr)\z^{(l)}  
    -  g_a^{(l)}\Q^{(l)\top} \n^{(l)},
\end{align}
where we have defined the variables $\n^{(l)} = (1 / g_a^{(l)}) \Q^{(l)} \z^{(l)}$ for each layer of the network.


The weight updates are derived via stochastic gradient descent of the loss given in Eq.~\eqref{eq:linear_final} after the neural dynamics have reached equilibrium. These are given by
\begin{subequations}\label{eq:learning}
\begin{align}
    \delta  \W_b^{(l)} &\propto
    \Big[g_a^{(l+1)} \bigl(\W_a^{(l+1)\top}\z^{(l+2)} - \Q^{(l+1)\top}\n^{(l+1)}\bigr) - c^{(l+1)}  {\z^{(l+1)}} \Big] {\z^{(l)\top}} \label{eq:W_b}\,,\\
    \delta  \W_a^{(l)} &\propto {\z^{(l+1)}} \z^{(l)\top} -  \W_a^{(l)}\label{eq:W_a}\,,\\
    \delta \Q^{(l)} &\propto \n^{(l)}\z^{(l)\top}-\Q^{(l)}\,.\label{eq:Q}
\end{align}
\end{subequations}
We used the neural dynamics equilibrium equation for $\z^{(l)}$ to simplify the weight update for $\W_b^{(l)}$. 
This yields our online algorithm (Alg.~\ref{alg:constraint_PC}), with the architecture shown in Fig.~\ref{fig:schematic}b. The algorithm can be implemented in a biologically plausible neural network as in Fig.~\ref{fig:network}; see Sec.~\ref{sec:bio}.

\begin{figure}[H]\vspace{-10pt}
\centering
\scalebox{0.90}{
  \begin{minipage}[t]{\textwidth}
\begin{algorithm}[H]\vspace{2pt}
  \caption{\hspace{-2.5pt}\textbf{: } Covariance Constrained predictive coding algorithm (BioCCPC)
}
  \label{alg:constraint_PC}
\begin{algorithmic}
  \State \textbf{input:} $\x_t$, $\y_t$ \Comment{new sample and previous weight matrices}\\\smallskip
  \hphantom{\textbf{input:}} $\W_a^{(l)}, \W_b^{(l-1)}, \Q^{(l)} \qquad \forall l \in \{1,\ldots,n-1\}$ 
  \smallskip\smallskip
  
  \State $\z_t^{(0)} \gets \x_t, \; \z_t^{(l)} \gets \W_b^{(l-1)}\z_t^{(l-1)}$ \Comment{initialize latents via forward pass}\smallskip
      \State  \textbf{run until convergence for all $l \in \{1,n-1\}$ }\Comment{neural dynamics}\smallskip\smallskip
      \State \qquad $\n^{(l)} \gets (1 / g_a^{(l)}) \Q^{(l)} \z_t^{(l)}$  
      \smallskip
      \State \qquad $\v_a^{(l)} \gets \W_a^{(l)\top}\z_t^{(l+1)} -  \Q^{(l)\top}\n^{(l)}$ \; , 
      %
      \qquad $ \v_b^{(l)} \gets \W_b^{(l-1)}\z_t^{(l-1)}$  
      \smallskip
      \State \qquad $ \z_t^{(l)} \gets \z_t^{(l)} + \tau^{-1}\left[-g_{\text{lk}}^{(l)}\z_t^{(l)}+g_{a}^{(l)}(\v_a^{(l)}-\z_t^{(l)})+g_{b}^{(l)}(\v_b^{(l)}-\z_t^{(l)}) \right]$  
      \bigskip
      \State \textbf{update  for all $l \in \{1,n-1\}$ }\Comment{synaptic weight updates} \smallskip\smallskip
      \State \qquad  $\W_a^{(l)} \gets \W_a^{(l)}+\eta\bigl(\z_t^{(l+1)} \z_t^{(l)\top} -  \W_a^{(l)} \bigr)$  
      \smallskip
      \State \qquad $\W_b^{(l-1)} \gets \W_b^{(l-1)}+ \eta \Big[g_a^{(l)} \bigl(\W_a^{(l)\top}\z_t^{(l+1)} - \Q^{(l)\top}\n^{(l)}\bigr) - c^{(l)}  {\z_t^{(l)}} \Big] {\z_t^{(l-1)\top}} $  
      \smallskip
      \State \qquad $\Q^{(l)} \gets \Q^{(l)}+\eta \bigl( \n^{(l)}\z_t^{(l)\top}-\Q^{(l)}\bigr) $  
      \bigskip
  \State \textbf{output:}  $\z_t^{(l)} \qquad \forall l \in \{0,\ldots,n\}$ \Comment{internal representations and updated weights}\smallskip\\
  \hphantom{\textbf{output: }}$\W_a^{(l)}, \W_b^{(l-1)}, \Q^{(l)} \qquad \forall l \in \{1,\ldots,n-1\}$ 
\end{algorithmic}
\end{algorithm}
\end{minipage}
}
\end{figure}

\section{Biological implementation and comparison with experimental observations}\label{sec:bio}

In this section, we introduce a biologically plausible neural circuit that implements the CCPC algorithm, denoted by BioCCPC. We also demonstrate that the details of this circuit resemble the neurophysiological properties of pyramidal cells in the neocortex and the hippocampus. 

\subsection{Neural architecture}

The algorithm for BioCCPC (Alg.~\ref{alg:constraint_PC}) summarized by the neural dynamics from Eq.~\eqref{eq:lin_neural_dynamics} and weight update rules from Eqs.~\eqref{eq:learning} can be implemented by a neural circuit with schematic shown in Fig.~\ref{fig:network}.

The activity of the $n-1$ hidden layers, $\{\z^{(l)}\}_{l=1}^{n-1}$, is encoded as the outputs of $n-1$ sets of neurons, representing pyramidal neurons of different cortical regions.  The matrices $\W_a^{(l)}$ and $\W_b^{(l)}$ are encoded as the weights of synapses between the pyramidal neurons of adjacent layers. Explicitly, the matrix $\W_{a}^{(l)}$ (resp.\ $\W_{b}^{(l-1)}$) is the efficacy of the synapse connecting $\z^{(l+1)}$ (resp.\ $\z^{(l-1)}$) to the pyramidal neurons $\z^{(l)}$. Because of the disjoint nature of these two inputs, we model these as synapsing respectively onto the distal (apical tuft) and proximal (mostly basal) dendrites of the pyramidal neurons, respectively; see~Fig.~\ref{fig:network}. This is reminiscent of cortical pyramidal neurons, which also have two integration sites, the proximal compartment comprised of the basal and proximal apical dendrites providing inputs to the soma, and the distal compartment comprised of the apical dendritic tuft~\cite{Larkum756,pyramidal_review}. These two compartments receive excitatory inputs from two separate sources~\cite{takahashimagee,Larkum2013}. 

Similarly, the auxiliary variables $\n^{(l)}$ are represented by the activity of interneurons in each cortical region.\footnote{There are multiple types of interneurons targeting pyramidal cells~\cite{klausberger2003brain,riedemann2019diversity}. The interneurons of BioCCPC most closely resemble the somatostatin-expressing interneurons, which preferentially inhibit the apical dendrites.}  The $\Q^{(l)}$ synaptic weights are encoded in the weights of synapses connecting $\n^{(l)}$ to $\z^{(l)}$, while $\Q^{(l)\top}$ models the weights of synapses from $\z^{(l)}$ to $\n^{(l)}$. In a biological setting, the implied equality of weights of synapses from $\z^{(l)}$ to $\n^{(l)}$ and the transpose of those from $\n^{(l)}$ to $\z^{(l)}$ can be guaranteed approximately by application of the same Hebbian learning rule~\cite{RRR}. However, note that our learning rules, unlike previous work, do not require interneurons of one layer to be connected to the pyramidal neurons of another layer.

\begin{figure}[t]
  \begin{center}
    \includegraphics[width=0.9\textwidth]{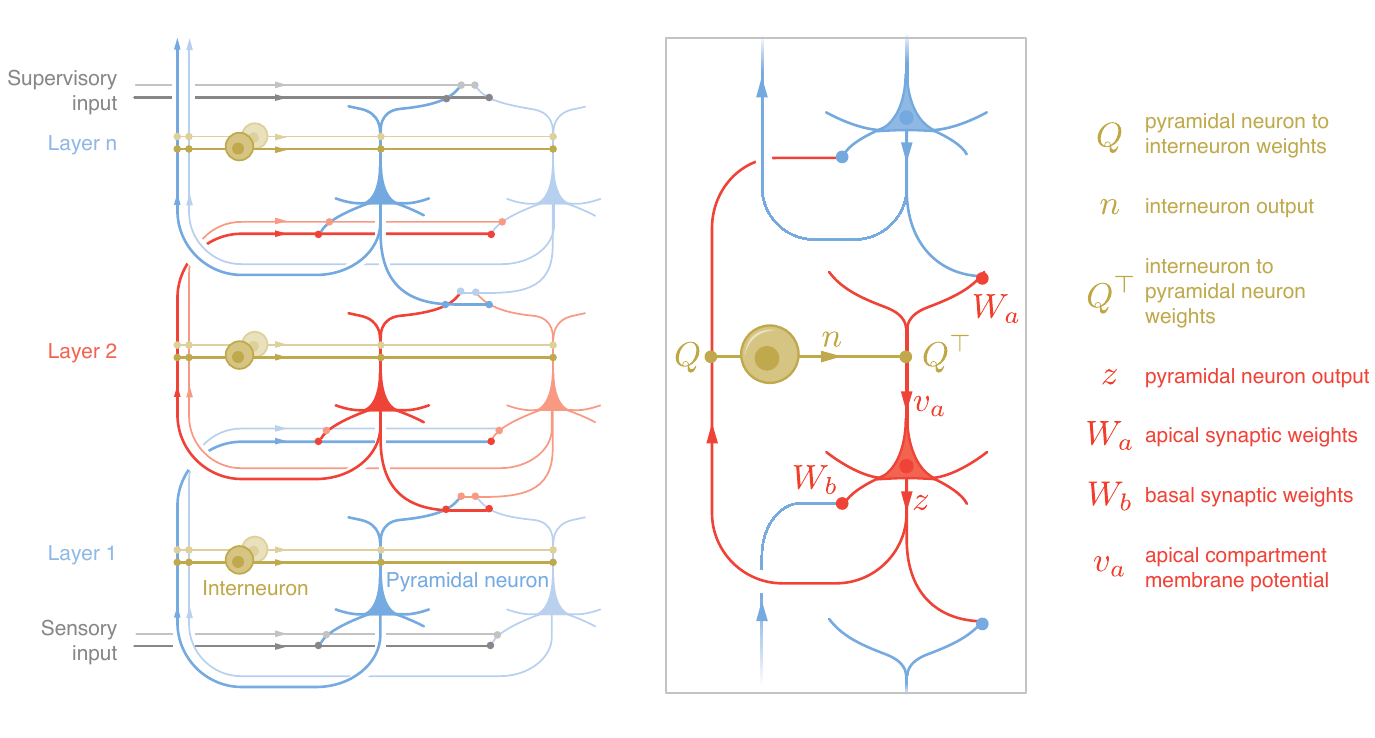}
  \end{center}
  \caption{Schematic of the biological implementation of the BioCCPC algorithm. Left: overall connectivity pattern, showing all-to-all connections between pyramidal neurons in consecutive layers, and between pyramidal and interneurons in the same layer. Right: zoom in on one neuron and its immediate neighbors, showing the membrane potentials and synaptic weights that are relevant for our algorithm.}
  \label{fig:network}
\end{figure}

\subsection{Neural dynamics and canonical components} 

The neural dynamics of the pyramidal neurons of our circuit given by Eq.~\eqref{eq:lin_neural_dynamics} can be recast as the dynamics of a three compartment neuron:
\begin{align}\label{eq:bio_neural_dynamics}
    \tau \dot \z^{(l)} =-g^{(l)}_{\text{lk}}\z^{(l)}+g_{a}^{(l)}(\v_a^{(l)}-\z^{(l)})+g_{b}^{(l)}(\v_b^{(l)}-\z^{(l)}),
\end{align}
where we have defined the apical and basal compartmental membrane potentials
\begin{equation}
    \v_a^{(l)}= \W_a^{(l)\top}\z^{(l+1)} -  \Q^{(l)\top}\n^{(l)}\;,\qquad
    \v_b^{(l)}=\W_b^{(l-1)}\z^{(l-1)}.
\end{equation}
Here we defined the leak conductance $g_\text{lk}^{(l)} = c^{(l)} - g_a^{(l)}$.

Interestingly, this is the same neural dynamics of a three-compartment neuron posited previously in~\cite{Urbanczik2014,Sacramento2018}, which we derived in a normative way. Here, $g_a$ (resp. $g_b$) are conductances between the apical (resp. basal) compartments and the soma, and $g_\text{lk}$ is the somatic leak. Our derivation clarifies the relationship between the physiological quantities and the objective function. Indeed, the basal and apical  conductances play the role of the inverse variances in the predictive coding objective. These inverse variances, generally referred to as `precisions' in the PC literature, have previously been argued to be encoded in lateral inhibitions~\cite{Friston2005} or top-down attention~\cite{Friston2010a,kanai2015cerebral}. In our model, they are simply encoded in the compartmental conductances. In Sec.~\ref{sec:theory} we show that the leak term also has a functional role and determines a threshold for dynamically learning more diverse latent variables~\cite{BioCCA,AdaDimRed}.

\subsection{Synaptic weight updates} 
Matching observations in the cortex, in our circuit~\cite{Larkum2013,gilbert2013top,keller2018predictive,Larkum756,pyramidal_review}, the basal and apical weights ($\W_b$ and $\W_a$) of our algorithm are updated differently. The basal synaptic weights of the pyramidal neurons, given by the elements of $\W_b^{(l)}$, are updated by the product of two factors represented in the corresponding post- and pre-synaptic neurons~(Eq.~\ref{eq:W_b}):
\begin{equation}\label{eq:simplified_dW_b}
     \delta \W_b^{(l-1)} \propto \Big[g_a^{(l)}\v^{(l)}_a-c^{(l)}\z^{(l)}\Big]\z^{(l-1)\top}\,.
\end{equation}
The first term in parentheses, $g_a^{(l)} \v^{(l)}_a$, is proportional to the total apical current, given by the difference between the excitatory synaptic current in the apical tuft, $\W_a^{(l)\top}\z^{(l+1)}$, and the inhibitory current induced by interneurons synapsing onto the distal compartment, $\Q^{(l)\top}\n^{(l)}$.
Biologically, this factor can be approximated by the calcium plateau potential traveling down the apical shaft. The calcium plateau potential has been experimentally seen to drive plasticity of the basal synapses matching the derived update rule of our circuit~\cite{Larkum2013,golding2002,bittner2015,bittner2017behavioral,MageeGrienberger2020}. 
The second factor is simply the neural output. 
Because the total update is not purely dependent on the action potentials of the pre- and post-synaptic neurons, such plasticity is called non-Hebbian~\cite{MageeGrienberger2020}.
%
%
The synaptic learning rule for the apical weights $\W_a^{(l)}$ (Eq.~\ref{eq:W_a}) and the synapses connecting the interneurons and the pyramidal neurons $\Q^{(l)}$ (Eq.~\ref{eq:Q}) are simply Hebbian, also matching experimental observations in the cortex~\cite{sjostrom2006,Kampa2007}.

\subsection{Features and limitations}

Here we summarize the main differences between the biological implementations of CCPC and PC. In PC, the forward and backward weights (in our notation $\W_b$ and $\W_a^\top$) are symmetric, whereas in CCPC, they are not constrained to be symmetric and are generically not so (see Sec.~\ref{sec:experiment}). Another difference is that in PC, the connection between value neurons and error neurons is one-to-one and fixed; this constraint does not exist between the pyramidal and interneurons of our model, which are no longer one-to-one, and have plastic connectivity. For the relationship between our interneurons and the PC error neurons see Sec.~\ref{sec:theory}.

While we have shown the improvement in biological realism of our model over the traditional PC network, this was done in the linear case. However, recent work on deep linear networks~\cite{saxe2013exact}, has provided many insights into the learning dynamics of deep networks. Some of the properties discovered in the deep linear network, like ``balancedness" of weights~\cite{arora2018optimization,du2018algorithmic}, generalize for certain nonlinear networks~\cite{du2018algorithmic}. We believe several of our observations will generalize in a similar fashion.

There are other aspects that can also be improved. For example, the connectivity between pyramidal neurons and interneurons of CCPC is required to be symmetric. This symmetry can be achieved via Hebbian learning rules~\cite{RRR}; however, it would be interesting to explore whether such symmetry is indeed required. Note that in CCPC, symmetric connectivity is only required between neurons of the \emph{same} cortical region, which is a much less stringent biological requirement than the that of symmetric weights between \emph{different} cortical regions, as in PC. 
Another shortcoming of our model is that while the teaching signal for the basal synapses in CCPC is signed and graded, in the cortex these signals are generally believed to be stereotypical~\cite{Larkum2013}. Graded calcium-mediated signals were recently observed \cite{Gidon2020}.


\section{Theoretical arguments}\label{sec:theory}

We summarize some of the theoretical features of our framework. For proofs see SM Sec.~\ref{app:proofs}.)

\paragraph{The errors are implicitly computed in $\v_a^{(l)}$.} In backpropagation as well as in PC, the learning algorithm computes a loss (or a local error) which is then used to compute updates to the weights. In CCPC, no error or loss is explicitly computed. So how does CCPC learn? In the following proposition, we show that in the $g_a\to 0$ limit, the algorithm implicitly computes an error in its weight updates. This quasi feed-forward or weak nudging limit, has been used to explore the learning dynamics of biologically plausible networks~\cite{Sacramento2018} and is related to the `fixed prediction' assumption ~\cite{Whittington2017,PCarbitraryPoor,PCexactBackprop}.

\begin{prop}\label{prop:forward_weights} Assume that the learning rules (Eqs.~\ref{eq:learning}) are at equilibrium and we receive a new datapoint given by $\x_{T+1}, \y_{T+1}$. For $c^{(l)}=0$ and in the limit of $\e\equiv g_a/g_b\to 0$, the leading term in $\e$ for the forward weight updates $\delta \W_b^{(l-1)}$ for this new sample is given by 
\begin{align*}
     \delta  \W_b^{(l-1)} \propto \v_{a,T+1}^{(l)}{\z_{T+1}^{(l-1)}}^\top = &\;   
    \e^{n-l}\left[{\W_a^{(l)\top}\cdots{\W_a^{(n-1)\top}}}(\y_{T+1} - \tilde \y_{T+1})\right]  {\z_{T+1}^{(l-1)\top}}+\mathcal{O}({\e^{n-l+1}}),
\end{align*}
where $\tilde \y_{T+1}=\Y\Z^{(l-1)\top}
        \left( \Z^{(l-1)}\Z^{(l-1)\top}\right)^{-1} \z_{T+1}^{(l-1)}$ is the optimal linear inferred value for $\y_{T+1}$.
\end{prop}
This implies that $\v_a^{(l)}$, which we argued can be implemented as the calcium plateau potential, implicitly encodes an error signal. This is in spirit similar to difference target propagation~\cite{Lee2014}, where the difference between the forward pass and the target is explicitly computed and backpropagated. 
\paragraph{Adaptive latent dimension discovery for a single module.} For optimization within a single hidden layer, we show that the system performs adaptive latent dimension discovery by thresholding eigenvalues of a Gramian including contributions from both adjacent layers.  If we take the loss function in \eqref{eq:linear_final}, setting $g_a=1,g_b=\epsilon$, and optimizing over $\W_a,\W_b$, we get the objective $\tr[ (c+\epsilon)\Z^\top \Z -\Z^\top\Z\big[\epsilon \X^\top (\X \X^\top)^{-1}\X+\Y^\top \Y\big]$ for $\Z$. The next theorem indicates the solution:

%

%
\begin{thm}\label{thm:pca}
Let the concatenated matrix $\bXi = [ \epsilon^{1/2}\C_x^{-1/2}\X,\Y]\in \R^{(d_X+d_Y)\times T}$ have SVD $\bXi=\sum_{\alpha=1}^{(d_X+d_Y)} \uvec_\alpha\lambda_\alpha \vvec_\alpha^\top$ with $\{\uvec_\alpha\in\R^{(d_X+d_Y)} \}$ and $\{\vvec_\alpha\in \R^T\}$ both being sets of orthonormal vectors, with the convention that $(\lambda_\alpha)$ are sorted in decreasing order.  We consider the minimization
\begin{equation}\label{eq:Simplified_PC}
\min_{\Z\in \R^{d\times T},\,\,\Z\Z^\top \preccurlyeq I_d}\tr\left[ (c+\epsilon)\Z^\top \Z -\Z^\top\Z\big[\epsilon \X^\top (\X \X^\top)^{-1}\X+\Y^\top \Y\big] \right].
\end{equation}
Then one of the optimal solutions
is given by $\hat \Z=\sum_{\alpha=1}^D \w_\alpha\mathbbm {1}(\lambda_\alpha>c+\epsilon)\vvec_\alpha^\top$ where $\{\w_\alpha\in\R^d \}$ is an arbitrary set of orthonormal vectors, with $D=\min(d,d_X+d_Y)$.
\end{thm}
Thus, the latent modes are found by thresholding the eignevalues~\cite{AdaDimRed} of the sum of Gramians~\cite{BioCCA} in Eq.~\eqref{eq:Simplified_PC}. This implies that the quantity $c^{(l)}$ that we showed is encoded in the somatic compartment leak conductance, biologically sets an adaptive threshold for the dimensionality of the latent representation.

From this theorem, we derive the computational limits of $\e\to\infty$ and $\e\to0$ in the supplementary materials. We also show that the neurons of the narrowest layer get whitened and the other layers become low-rank. For further details of these statements and all proofs see SM Sec.~\ref{app:proofs}.

\section{Numerical experiments}\label{sec:experiment}
\begin{figure}[!tbh]
\vspace{-5pt}
    \center\includegraphics{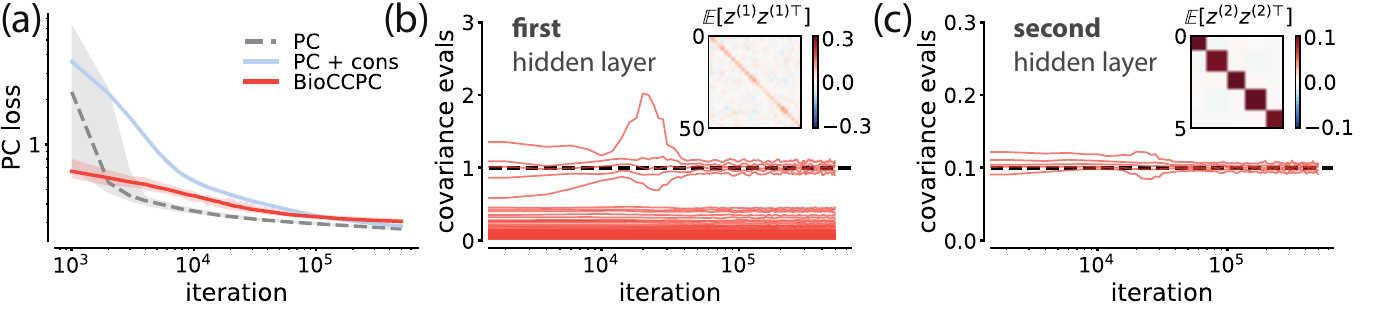}
    \caption{Our algorithm (BioCCPC) compares well to less biologically plausible predictive-coding (PC) implementations, while enforcing an inequality constraint on the covariance of hidden-layer activations. Simulations run on the MNIST dataset, using a network with two hidden layers, sizes 50 and 5. Shaded area represents the 95\% range out of 40 runs. (a) Evolution of the predictive-coding loss, Eq.~\eqref{eq:obj_doubled}, on a validation set during learning. 
    (b), (c) Evolution of the eigenvalues of the covariance matrix of hidden-layer activations in the two hidden layers. The dashed lines show the scale of the constraint; see SM Sec.~\ref{app:upperbound}. Inset: covariance matrix at end of training. Note how activity is whitened in the second hidden layer (c), and low-rank in the first hidden layer (b). \label{fig:learning_curves}}
\end{figure}

We tested our BioCCPC algorithm on the MNIST dataset~\cite{lecun1998gradient}, with one-hot labels. We centered the input images  and normalized by the standard deviation. We compared our results with a standard predictive coding network~\cite{Whittington2017} with linear activation (PC) and with a version that includes a covariance constraint (PC+cons; Sec.~\ref{app:supp_exp}). We used unit variances, $\sigma^{(l)2} = 1$, in the predictive-coding networks, and matching conductances in our algorithm, $g_b^{(1)} = g_a^{(n)} = 1$, $g_a^{(l)} = g^{(l)} = 1/2$ for other $l$.

Our method finds solutions that are almost optimal in terms of predictive-coding loss (see Fig.~\ref{fig:learning_curves}a) and are generally better than the solutions obtained from simply adding the covariance constraint to a predictive-coding network. Learning is slower in our network than in unconstrained PCN, as it takes time for the covariance constraint to be enforced. It is possible that different initialization schemes could improve the learning speed, but we leave this for future work.

The covariance constraint is saturated in the narrowest layer---in our example, the second hidden layer, of size 5---see Fig.~\ref{fig:learning_curves}c. The activation in the other hidden layer is rank-restricted by the narrowest layer, as described above, and so it is only whitened within a 5D subspace, Fig.~\ref{fig:learning_curves}b.
In the supplementary materials, we provide further numerical experiments, where we look at the effects of the thresholding, the size of the network, and performance on other datasets. We also show that the forward and backward weights, $\W_b$ and $\W_a^\top$, are not symmetric, in accordance with biological observations, and we explore their relation. See SM Sec.~\ref{app:supp_exp} for these and other simulation details. 


\section{Conclusion}

We have derived a biologically plausible algorithm for covariance-constrained predictive coding and have shown that it avoids the criticisms of PC and has many features in common with recent experimental observations. Our algorithm does not have symmetric forward and backward connectivity matrices, and does not need one-to-one connectivity between any neurons. Furthermore, we showed that the learning rules of our algorithm closely match experimental observations, and we connected our parameters to neurophysiological quantities. Using the simplicity afforded to us because of the linearity, we showed a number of interesting properties of our algorithm, including the effect of the somatic leak term, which acts as an adaptive thresholding mechanism. We also showed how our algorithm implicitly computes a difference term that it uses for learning. We argued that this term could be encoded in the calcium plateau potential, which is known to affect the learning of basal synapses. We hope that this concrete connection between our normative model and observation can further our understanding of cortical computation by guiding future experiments.

\section*{Broader impact}

We do not foresee any potential societal harm caused by our work. We hope our contribution moves forward the discussion on the relationship between multi-layer networks and cortical structure.
\bibliography{biblio.bib}
\bibliographystyle{unsrt}

\clearpage

\appendix

\begin{center}\LARGE{\textbf{Supplementary Materials}}\end{center}
\vspace{3pt}

\section{Derivation of upper bound}\label{app:upperbound}
In this section we give a more detailed derivation of our framework, using a slightly modified order of steps compared to the main text. This should help clarify some aspects of the final algorithm.

Similar to the main text, we start with  the predictive coding objective and introduce the weight doubles. We arrive at
\begin{multline}
    \label{eq:app_deep_paired}
    \min_{\Z,\W} L = \min_{\Z,\W_a,\W_b} \frac 12\sum_{l=1}^{n-1} \left[g_b^{(l)} \bigl\|\Z^{(l)}-\W_b^{(l-1)} \Z^{(l-1)}\bigr\|_F^2 \right.\\ \left. +\,g_a^{(l+1)}\bigl\|\Z^{(l+1)}-\W_a^{(l)} \Z^{(l)}\bigr\|_F^2+c^{(l)}\bigl\|\Z^{(l)}\Z^{(l)\top}\bigr\|^2_F\right],
\end{multline}
where we have also introduced a quadratic cost (a prior on $\Z$) with coefficient $c^{(l)}$. We now impose the constraint  on the empirical covariance, $\frac1T\Z^{(l)} \Z^{(l)\top} \preceq \rho^{(l)} \I$, of the internal variables $\Z^{(l)}$\footnote{The factors $\rho^{(l)}$ which set the scale of the constraint were set to 1 in the derivation in the main text for brevity. However, it is sensible for this scale factor to vary from layer to layer and possibly depend on the layer width. In Sec.~\ref{app:supp_exp}, we explore the effects of this parameter on the performance of the network.}:
\begin{multline}
    \label{eq:app:minL_with_cons}
    \min_{\Z,\W} L  \leq\min_{\substack{\Z,\W_a,\W_b\\\tfrac1T \Z^{(l)} \Z^{(l)\top} \preceq \rho^{(l)}\I}} \frac 12\sum_{l=1}^{n-1} \left[g_b^{(l)} \bigl\|\Z^{(l)}-\W_b^{(l-1)} \Z^{(l-1)}\bigr\|_F^2 \right.\\ \left. +\,g_a^{(l+1)}\bigl\|\Z^{(l+1)}-\W_a^{(l)} \Z^{(l)}\bigr\|_F^2+c^{(l)}\bigl\|\Z^{(l)}\Z^{(l)\top}\bigr\|^2_F\right].
\end{multline}
Expanding this and using the constraint we again upper bound the objective:
\begin{multline}
    \min_{\Z,\W} L  \leq\min_{\substack{\Z,\W_a,\W_b\\\tfrac1T \Z^{(l)} \Z^{(l)\top} \preceq \rho^{(l)}\I}} \frac 12\sum_{l=1}^{n-1} \left[g_b^{(l)} \bigl\|\Z^{(l)}-\W_b^{(l-1)} \Z^{(l-1)}\bigr\|_F^2 + g_a^{(l+1)}\bigl\|\Z^{(l+1)}\Z^{(l+1)\top}\bigr\|^2_F
    \right.\\ \left. 
    + \,g_a^{(l+1)}   \tr\bigl(- 2 \Z^{(l+1)\top}\W_a^{(l)}\Z^{(l)} +T\rho^{(l)} \W_a^{(l)\top} \W_a^{(l)}\bigr)+c^{(l)}\bigl\|\Z^{(l)}\Z^{(l)\top}\bigr\|^2_F\right],
\end{multline}
We next implement the inequality constraint using positive-definite Lagrange multipliers:
\begin{align}\label{app:Lsum}
    \min_{\Z,\W} L  \leq& \min_{\Z,\W_a,\W_b}\max_\Q \sum_{l=1}^{n-1} \tilde L^{(l)},
\end{align}
where 
\begin{equation}\label{app:Lsum_term}
    \begin{split}
      \tilde L^{(l)} &= \frac 12 \left[g_b^{(l)} \bigl\|\Z^{(l)}-\W_b^{(l-1)} \Z^{(l-1)}\bigr\|_F^2 
      +c_q^{(l)} \tr \Q^{(l)\top}{\Q^{(l)}}\bigl({\Z^{(l)}}\Z^{(l)\top} - T\rho^{(l)}\I\bigr)
    \right.\\
     &\qquad\qquad \left. +\, g_a^{(l+1)}   \tr\bigl(- 2 \Z^{(l+1)\top}\W_a^{(l)}\Z^{(l)} +T\rho^{(l)} \W_a^{(l)\top} \W_a^{(l)}\bigr) \right.\\
     &\qquad\qquad \left. +\,c^{(l)}\bigl\|\Z^{(l)}\Z^{(l)\top}\bigr\|^2_F+g_a^{(l+1)}\bigl\|\Z^{(l+1)}\Z^{(l+1)\top}\bigr\|^2_F\right],
    \end{split}
\end{equation}
where we have added a positive constant $c_q^{(l)}$ defining the coefficient of the inequality constraint. This constant does not affect the solution of the optimization problem but does affect the details of the neural dynamics. Motivated by biological arguments\footnote{Since the feedback of the interneurons is received and integrated in the apical compartment, we expect the apical-to-soma conductance to also be present here.}, we set this constant equal to $g_a^{(l)}$. In doing so, the interneuron variables no longer need a $(1 / g_a^{(l)})$ in their definition, as in the main text, and are instead given by  $\n^{(l)} = \Q^{(l)} \z^{(l)}$.

A final upper bound to \eqref{app:Lsum} can be found if each $\tilde L^{(l)}$ is optimized individually. Explicitly, we perform gradient descent of each variable, $\Z^{(l)}, \W_a^{(l)}, \W_b^{(l)}$, only with respect to $ \tilde L^{(l)}$ and not with respect to the entire objective $L$:
\begin{equation}
\label{app:splitgrads}
    \begin{split}
    \delta \Z^{(l)}= \nabla _{\Z^{(l)}} L &\;\to\;  \nabla _{\Z^{(l)}} \tilde L^{(l)},\\
    \delta \W_a^{(l)}= \nabla _{\W_a^{(l)}} L &\;\to\;  \nabla _{\W_a^{(l)}} \tilde L^{(l)}, \\
    \delta \W_b^{(l)}= \nabla _{\W_b^{(l)}} L &\;\to\;  \nabla _{\W_b^{(l)}} \tilde L^{(l)}.
    \end{split}
\end{equation}
This affords us an upper bound because we are no longer finding the actual minimum of $L$. Symbolically this can be written as in~\eqref{eq:linear_final}:
\begin{equation}
    \min_{\Z,\W} L  \leq  \sum_{l=1}^{n-1} \min_{\Z,\W_a,\W_b}\max_\Q \tilde L^{(l)},
\end{equation}
where the precise definition of this equation is given by the update equations~\eqref{app:splitgrads}. Note that these updates are no longer curl-free and hence cannot be found as the gradient of a single objective function.

\section{Theoretical section proofs}\label{app:proofs}
Note that the definition of $\e$ in Proposition~\ref{prop:forward_weights} is different from the definition in other sections.



\subsection{Proof of Proposition \ref{prop:forward_weights}}

To prove Proposition~\ref{prop:forward_weights}, we first show that the latent variables $\z^{(l)}$ have a specific form. As in the main text, we define $\e^{(l)}=g_a^{(l)}/g_b^{(l)}$ and, for simplicity, we assume $\e^{(l)}=\e$ is the same for all layers. For simplicity also, we assume $c^{(l)}$ (equivalently $d^{(l)}$ in previous section) is equal to zero, and use $c_q^{(l)}=g_a^{(l)}$ as in the previous section. This choice makes the notation simpler without changing our results. Equation~\eqref{eq:lin_neural_dynamics} becomes
\begin{align}\label{eq:neuralnoc}
    \dot \z^{(l)} = g_b^{(l)} \W_b^{(l-1)}\z^{(l-1)}+g_a^{(l)} \W_a^{(l)\top}\z^{(l+1)} -  g_b^{(l)}\z^{(l)}  
    -  g_a^{(l)}\Q^{(l)\top}\Q^{(l)} \z^{(l)}.
\end{align}

\begin{prop}\label{prop:latentsol}
Assume the neural dynamics given in Eq.~\eqref{eq:neuralnoc} are at equilibrium. Then the latent variables take the form 
\begin{equation}\label{eq:latentsol}
    \z^{(n-k)}=\B_k\left(\W_b^{(n-k-1)}\z^{(n-k-1)}+\e^k \A_k{\W_a^{(n-1)\top}} \y\right),
\end{equation}
where $\B_k$ and $\A_k$ are independent of $y$ and defined recursively via
\begin{align*}
    \A_{k+1} &= {\W_a^{(n-k-1)\top}} \B_k \A_k,\\
    \B_{k+1} &=\left(1-\e \F^{(n-k-1)}{\W_a^{(n-k-1)\top}} \B_k\W_b^{(n-k-1)}\right)^{-1}\F^{(n-k-1)},
\end{align*}
with $\F^{(l)}=\left(\I+\e {\Q^{(l)\top}}\Q^{(l)}\right)^{-1}$ and
\begin{equation*}
    \B_1 = \F^{(n-1)} \;, \quad \A_1 =1.
\end{equation*}
\end{prop}
\begin{proof}
We prove this via induction. Starting from the neural dynamics Eq.~\eqref{eq:neuralnoc}, at equilibrium ($\dot \z^{(l)}=0$) we have:
\begin{align*}
    \z^{(l)}= \F^{(l)} \left(\W_b^{(l-1)}\z^{(l-1)} + \e{\W_a^{(l)\top}} \z^{(l+1)}\right).
\end{align*}
Setting $l=n-1$ we get:
\begin{align*}
    \z^{(n-1)}= \F^{(n-1)} \left(\W_b^{(n-2)}\z^{(n-2)} + \e{\W_a^{(n-1)\top}} \y\right),
\end{align*}
which satisfies the proposed form in Eq.~\eqref{eq:latentsol} for $k=1$. Now, we assume that the form in Eq.~\eqref{eq:latentsol} holds for $z^{(n-k)}$, and show that it also holds for $z^{(n-k-1)}$. We have
\begin{align*}
    \z^{(n-k-1)}=& \F^{(n-k-1)}\left(\W_b^{(n-k-2)}\z^{(n-k-2)} + \e{\W_a^{(n-k-1)\top}} \z^{(n-k)}\right)\\
    =& \F^{(n-k-1)}\left(\W_b^{(n-k-2)}\z^{(n-k-2)} + \e{\W_a^{(n-k-1)\top}} \B_k\left(\W_b^{(n-k-1)}\z^{(n-k-1)}+\e^k \A_k{\W_a^{(n-1)\top}} \y\right)\right).
\end{align*}
Collecting $\z^{(n-k-1)}$ and multiplying by the inverse on the left we arrive at:
\begin{align*}
    \z^{(n-k-1)}=&\left(1-\e \F^{(n-k-1)}{\W_a^{(n-k-1)\top}} \B_k\W_b^{(n-k-1)}\right)^{-1}\F^{(n-k-1)} \times \\
    &\hspace{115pt} \times\left(\W_b^{(n-k-2)}\z^{(n-k-2)} + \e^{k+1}{\W_a^{(n-k-1)\top}} \B_k \A_k{\W_a^{(n-1)\top}} \y\right),
\end{align*}
as proposed.
\end{proof}
Note that, whereas Proposition~\ref{prop:forward_weights} holds only in the $\e\to0$ limit, Proposition~\ref{prop:latentsol} holds for general $\e$. Using this result we show that the relationship between the activity, feedback and feedforward inputs of the $l^\text{th}$ module has a simple form given in the following corollary:
\begin{cor}\label{corr:leadingsol}
The leading order in $\e$ of the dependence of $\z^{(l)}$ and $\z^{(l+1)}$ on $\z^{(l-1)}$ and $\y$, which has non-zero dependence on $\y$, is given by
\begin{align*}
    \z^{(l)}
    &= \C_1^{(l)}\z^{(l-1)}
    +\mathcal{O}(\e^{n-l}),\\
    \z^{(l+1)}
    &= \C_2^{(l)}\z^{(l-1)} +\e^{n-l-1}{\W_a^{(l+1)\top}}\cdots{\W_a^{(n-1)\top}} \y+\mathcal{O}(\e^{n-l}),
\end{align*}
where $\C_1^{(l)}$ and $\C_2^{(l)}$ are respectively the expansions of $\B_{n-l}\W_b^{(l-1)}$ and $\B_{n-l-1}\W_b^{(l)}\B_{n-l}\W_b^{(l-1)}$ up to order $\e^{n-l-1}$.
\end{cor}
\begin{proof}
The result follows directly from applying Eq.~\eqref{eq:latentsol} twice and then dropping all but the leading dependence on $\e$ in the coefficient of the $\y$ term.
\end{proof}

Now, we restate and prove Proposition~\ref{prop:forward_weights}, which can be viewed  as a generalization of the results in~\cite{RRR}. Note that in this proposition, the exact form of $\C_1^{(l)}$ and $\C_2^{(l)}$ is not important. All that matters is that both $\z^{(l)}$ and $\z^{(l+1)}$ have a linear dependence on $\z^{(l-1)}$. Since the dependence of $\z^{(l)}$ on $\y$ is subleading compared to the dependence of $\z^{(l+1)}$ on $\y$, this can be dropped in the limit of~$\e\to0$.

\textbf{Proposition 1.} \textit{ Assume that the learning rules (Eqs.~\ref{eq:learning}) are at equilibrium and we receive a new datapoint given by $\x_{T+1}, \y_{T+1}$. For $c^{(l)}=0$ and in the limit of $\e\equiv g_a/g_b\to 0$, the leading term in $\e$ for the forward weight updates $\delta \W_b^{(l-1)}$ for this new sample is given by 
\begin{align*}
     \delta  \W_b^{(l-1)} \propto \v_{a,T+1}^{(l)}{\z_{T+1}^{(l-1)\top}} = &\;   
    \e^{n-l}\left[{\W_a^{(l)\top}\cdots{\W_a^{(n-1)\top}}}(\y_{T+1} - \tilde \y_{T+1})\right]  {\z_{T+1}^{(l-1)\top}}+\mathcal{O}({\e^{n-l+1}}),
\end{align*}
where $\tilde \y_{T+1}=\Y\Z^{(l-1)\top}
        \left( \Z^{(l-1)}\Z^{(l-1)\top}\right)^{-1} \z^{(l-1)}_{T+1}$ is the optimal linear inferred value for $\y_T$.
}

\begin{proof}The statement follows from two applications of the update rule for $\W_b$ given in Eq.~\eqref{eq:W_b} for $\W_b^{(l-1)}$. Putting $c^{(l)}=0$ we have:
\begin{equation}
    \delta  \W_b^{(l-1)} \propto
     \e \bigl(\W_a^{(l)\top}\z^{(l+1)} - \Q^{(l)\top}\Q^{(l)}\z^{(l)}\bigr) {\z^{(l-1)\top}}.
\end{equation}
The equilibrium condition given by $\delta  \W_b=0$ is given by averaging over the entire dataset:
\begin{equation}
    \Q^{(l)\top}\Q^{(l)}\Z^{(l)}\Z^{(l-1)\top} = \W_a^{(l)\top}\Z^{(l+1)}\Z^{(l-1)\top}.
\end{equation}
Now we plug the derived values in Corollary~\ref{corr:leadingsol}. After right-multiplying by $(\Z^{(l-1)}\Z^{(l-1)\top})^{-1}$ we get:
\begin{equation}
    \Q^{(l)\top}\Q^{(l)}\C_1 = \W_a^{(l)\top}\C_2 +  \e^{n-l-1}\W_a^{(l)\top}\cdots{\W_a^{(n-1)\top}} \Y\Z^{(l-1)\top}
    \big(\Z^{(l-1)}\Z^{(l-1)\top}\big)^{-1}.
\end{equation}
    Now, we look at the update for a new sample $T+1$. Plugging the above as well as the equations from Corollary~\ref{corr:leadingsol} into the update rule we see that the $\C_2^{(l)}$ dependent term cancels and we are only left with the term depending on $\y$:
\begin{align*}
    \delta  \W_b^{(l-1)} \propto &
     \e^{n-l}\bigl(\W_a^{(l)\top}\z_{T+1}^{(l+1)} - \Q^{(l)\top}\Q^{(l)}\z_{T+1}^{(l)}\bigr) {\z_{T+1}^{(l-1)\top}}\\
     =&\;   
    \e^{n-l}\left[{\W_a^{(l)\top}\cdots{\W_a^{(n-1)\top}}}(\y_{T+1} - \tilde \y_{T+1})\right]  {\z^{(l-1)\top}},
\end{align*}
with $\tilde \y_{T+1}$ as defined in the proposition.
\end{proof}

\subsection{Proof of Theorem~\ref{thm:pca} and consequences}

In this section we look at some properties of the within-layer computations of CCPC. We will focus on the computations within the $l^\text{th}$ layer and for clarity of notation, for this subsection only, we rename $\Z^{(l)}\to\Z$,  $\Z^{(l+1)}\to\Y$, and  $\Z^{(l-1)}\to\X$. We will assume $\rho^{(l)}=1$ for brevity. The results have straight-forward generalizations for the case where $\rho^{(l)}\neq 1$. 

\textbf{Theorem 1.} \textit{
Let the concatenated matrix $\bXi =\U\boldsymbol\Lambda \V^\top= [ \e^{1/2}\C_x^{-1/2}\X,\Y]\in \R^{(d_X+d_Y)\times T}$ have singular value decomposition $\bXi=\sum_{\alpha=1}^{(d_X+d_Y)} \uvec_\alpha\lambda_\alpha \vvec_\alpha^\top$, with $\{\uvec_\alpha\in\R^{(d_X+d_Y)} \}$ and $\{\vvec_\alpha\in \R^T\}$ both being sets of orthonormal vectors, and singular values $(\lambda_\alpha)$ sorted in decreasing order.  We consider the minimization
\begin{equation}\label{eq:Simplified_PC_2}
\min_{\substack{\Z\in \R^{d\times N}\\\Z\Z^\top \preccurlyeq I_d}}\tr\left[ (c+\e)\Z^\top \Z -\Z^\top\Z\big[\e \X^\top (\X \X^\top)^{-1}\X+\Y^\top \Y\big] \right].
\end{equation}
Then one of the optimal solutions
is given by $\hat \Z=\sum_{\alpha=1}^D \w_\alpha\mathbbm {1}(\lambda_\alpha>c+\e)\vvec_\alpha^\top$ where $\{\w_\alpha\in\R^d \}$ is an arbitrary set of orthonormal vectors, with $D=\min(d,d_X+d_Y)$.
}
\begin{proof}
Let the SVD of $\Z$ be $\W\boldsymbol\Lambda_\Z\V_\Z^\top$. Minimizing over $\Z$ amounts to finding the matrices $\W$, $\boldsymbol\Lambda_\Z$, and $\V_\Z$ that give the smallest  value of the objective.  Because of the constraint $\Z\Z^\top \preccurlyeq I_d$, $\boldsymbol\Lambda_\Z$ satisfies $0\le \lambda_{\Z\alpha}\le 1$. Plugging this form into the optimization objective we have
\begin{equation}\label{eq:objective}
\tr\left[ (c+\e)\Z^\top \Z -\Z^\top\Z\big[\e \X^\top (\X \X^\top)^{-1}\X+\Y^\top \Y\big] \right]=\tr\left[(c+\e)\boldsymbol\Lambda_\Z^\top\boldsymbol\Lambda_\Z-\boldsymbol\Lambda_\Z^\top\boldsymbol\Lambda_\Z\tilde \V^\top\boldsymbol\Lambda^\top\boldsymbol\Lambda\tilde \V^\top\right]\notag
\end{equation}
where $\tilde \V=\V^\top \V_\Z$. Note that $\W$ does not appear in this expression and therefore remains undetermined in this optimization objective. 

$\V$ only appears in the second term and hence can be found by optimizing the second term alone. This term is of the form $-\sum_{\alpha,\beta}\lambda_{\Z\alpha}^2|\tilde V_{\alpha\beta}|^2\lambda_{\beta}^2$. The $\tilde \V$ that minimizes the term is the one that pairs  $\lambda_{\beta}$ and $\lambda_{\Z\alpha}$, in the sorted order. Assuming $(\lambda_{\Z\alpha})$ are also sorted in the decreasing order, this is achieved by setting $\V_\Z=\V$.

Plugging this in, the entire expression simplifies to $$-\sum_\alpha (\lambda_\alpha-c-\e)\lambda_{\Z\alpha}.$$
Since $0\le \lambda_{\Z\alpha}\le 1$, it is clear that, for a minimum,  if $(\lambda_\alpha-c-\e)<0$, $\lambda_{\Z\alpha}=0$, while $(\lambda_\alpha-c-\e)>0$ implies $\lambda_{\Z\alpha}=1$. In case of equality, the choice of $\lambda_{\Z\alpha}$ does not matter. Hence $\lambda_{\Z\alpha}=\mathbbm {1}(\lambda_\alpha>c+\e)$ is a valid minimum. The vectors $\{\w_\alpha\}$ are the columns of $\W$, while $\{\vvec_\alpha\}$ are the columns of $\V$. Hence, the result.
\end{proof}

From this theorem, we derive two corollaries, describing the different limits of $\e\to0$ and $\e\to\infty$. The proofs of these corollaries follow directly from Theorem~\ref{thm:pca} which maps the within-layer computation of CCPC onto a PCA problem, and the known properties of PCA.

\begin{cor}\label{prop:EV_regimes_MSE}
If the input sample covariance matrix, $\C_x$, is full rank, the internal representation $\Z$ in the limit $\e \rightarrow 0$ is driven by the principal components (scores) of $\Y$, i.e.,
\begin{align}\label{prop_eq:mse}
\Z^{*} = \W_a^\top \Y + \e \W_b \X ,\end{align}
with $\W_a$ the top principal directions of $\Y$ and $\W_b$ the maximally correlated directions between $\X$ and the columns of $\W_a$. 
\end{cor}

\begin{cor}\label{prop:EV_regimes_CCA}
If the input sample covariance matrix, $\C_x$, is full rank, the internal representation $\Z$ in the limit $\e \rightarrow \infty$ is driven by the maximally correlated directions between $\C_x^{-1/2}\X$ and $\Y$, i.e.,
\begin{align}\label{prop_eq:cca}
\Z^{*} = \W_a^\top \Y + \W_b\X ,\end{align}
with $\W_a, \W_b$ projecting onto a common lower-dimensional subspace so that the projections are maximally correlated. 
\end{cor}

\subsection{Rank structure of CCPC}
In this section again we set $c^{(l)}=0$ and $\rho^{(l)}=1$, and assume $g_a=1,g_b=\epsilon$ is the same for all layers. The generalization to other cases is straightforward. We demonstrate a series of rank constraints on the intermediate layers of CCPC. The first of these is very general and requires no assumptions. The following statements become gradually more tight but also require more assumptions.

In this section, we define $r_x$ and $r_y$ to be the ranks of the the covariances of $\X$ and $\Y$, and $r^{(l)}$ to be the rank of the covariance of $\Z^{(l)}$ at the equilibrium of the CCPC update equations. We take $r^{(0)}=r_x$ and $r^{(n)}=r_y$.

\begin{cor}
The rank of the $l^\text{th}$ layer is bounded by $r^{(l+1)}+r^{(l-1)}$:
\begin{equation}
    \forall l\; : \; r^{(l)}\leq r^{(l+1)}+r^{(l-1)}.
\end{equation}
\end{cor}
\begin{proof}
This follows directly from Theorem~\ref{thm:pca} which describes $Z^{(l)}$ as the whitened PCA of the concatenation of $\Z^{(l-1)}$ and $\Z^{(l+1)}$.
\end{proof}

When the adjacent layers are whitened, i.e., the non-zero eigenvalues of $Z^{(l+1)}Z^{(l+1)\top}$ and $Z^{(l-1)}Z^{(l-1)\top}$ are all one, we can prove a tighter form of the rank constraint. This is relevant for the intermediate layers of CCPC other than the final hidden layer. This proposition is relevant for the final hidden layer only if the supervisory input $\Y$ is white.

\begin{prop}
Let $r_c^{(l)}$ be the rank of the covariance matrix $\Z^{(l-1)}\Z^{(l+1)\top}$. The rank of intermediate layers $l\leq n-2$ is given by:
\begin{equation}\label{eq:rank}
    \forall l\leq n-2 \; : \; r^{(l)} \leq 
        \begin{cases}
        r_c^{(l)}\leq \min \{r^{(l+1)},r^{(l-1)}\} & \e\geq1, \\
        r^{(l+1)} & \e<1 .
        \end{cases}
\end{equation}
\end{prop}
\begin{proof}
This is a slight generalization of~\cite{BioCCA}. It follows from the eigenvalue structure of the concatenated data matrix given in Theorem~\ref{thm:pca}, where only a subset of the eigenvalues exceed the threshold given by $\e$. The number of eigenvalues that exceed this threshold is given by Eq.~\eqref{eq:rank}.
\end{proof}

Applying this proposition to all the layers of a CCPC network, we arrive at the following corollary:
\begin{cor}
\label{app:cor:rank_bioccpc}
Assume that $\Y$ is white and $\e\geq 1$. In a CCPC network with $c^{(l)}\geq 0$, $\rho^{(l)}=1$, and layer widths $w^{(l)}$, the rank of the latent variables $\z^{(l)}$ is less than or equal to $\min \{r_x,r_y,w^{(1)},\cdots,,w^{(n-1)}\}$.
\end{cor}
In other words, the rank of the latent variables is bounded by the minimum of the input/output ranks and the bottleneck of the network.

\subsection{Alignment and relation to CCA}

In the limit where the internal loss terms in Eq.~\eqref{eq:obj_doubled} are zero, we effectively have a deep linear network~\cite{saxe2013exact}. A lot is known about deep linear networks, trained by gradient descent of weights. One of the remarkable properties of such networks is that the weights remain balanced, namely $\W^{(l)\top}\W^{(l)}-\W^{(l-1)}\W^{(l-1)\top}$ remains a constant, under gradient descent. In the case this difference is zero, the left singular vectors of $\W^{(l-1)}$ match the right singular vectors of $\W^{(l)}$~\citeSM{ji2019gradient}. In addition, for two class linear discriminants, $\W^{(l-1)}$ ultimately become rank 1, under gradient descent~\citeSM{ji2019gradient}.


\section{Numerical simulations}\label{app:supp_exp}

\begin{figure}[!tbh]
    \center\includegraphics{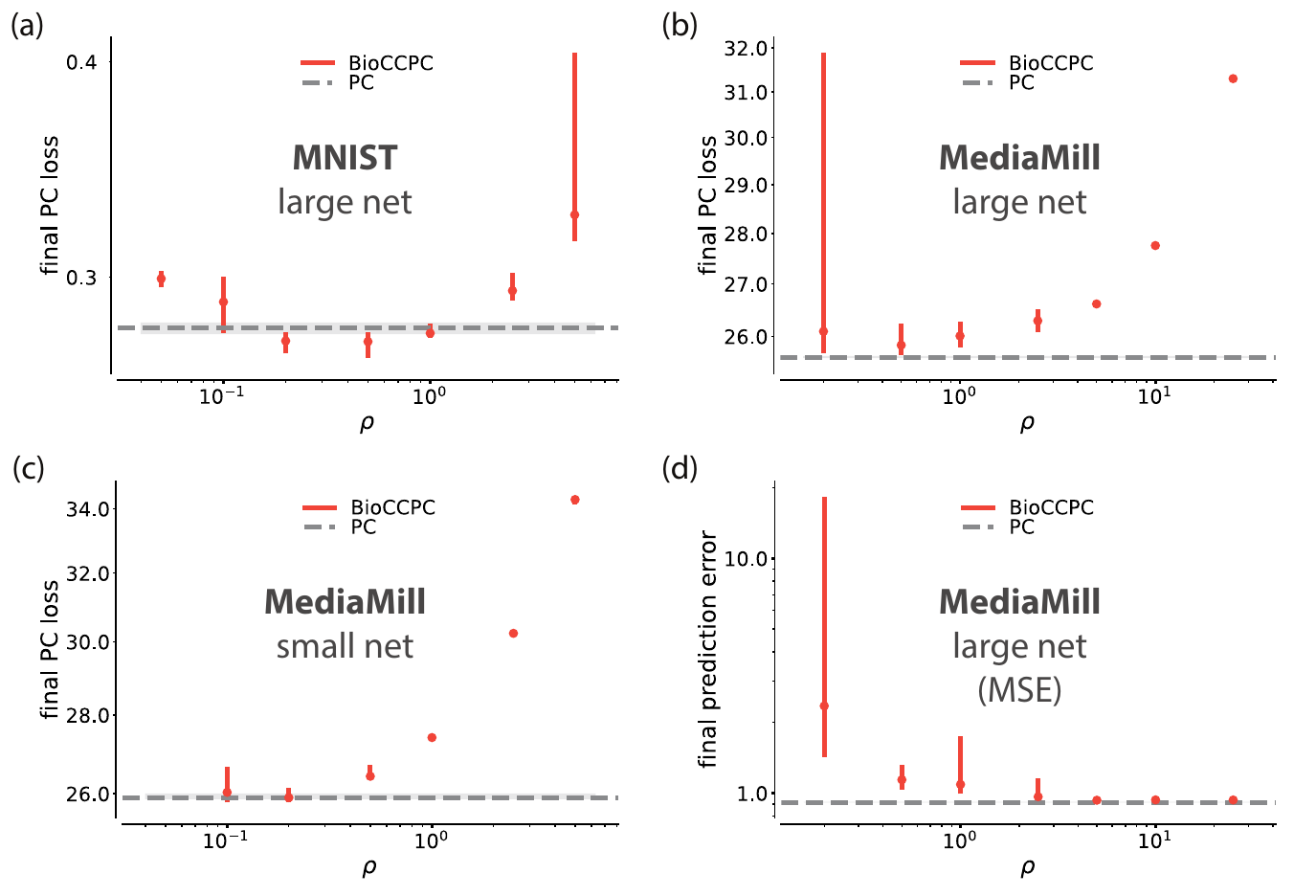}
    \caption{Dependence of BioCCPC performance on constraint scale. Unless otherwise noted, this is the predictive-coding loss (Eq.~\ref{eq:obj_doubled}) after 500,000 iterations. (a) A large network with two hidden layers of sizes 50 and 5 trained on MNIST. (b) The same large network as in (a) run on the Mediamill dataset. See text. (c) A small net (single hidden layer of size 5) on Mediamill. (d) Mean-squared prediction error for the large net on Mediamill, showing that the best prediction performance may occur at a different choice of $\rho$ compared to the lowest predictive-coding loss. \label{fig:app:constraint_scale}}
\end{figure}

\subsection{Effect of constraint scale}
The BioCCPC network uses an inequality constraint on the covariance of the hidden variables, $\frac 1T \Z^{(l)} Z^{(l)\top} \preceq \rho^{(l)}\I$, and this constraint is saturated at convergence.\footnote{More precisely, the constraint is saturated only within the subspace that $\Z^{(l)}$ occupies. In general, this is only the case in the layer with the smallest dimension, as the activity in the other layers is restricted to a subspace of dimension given by the smallest layer; see Corollary~\ref{app:cor:rank_bioccpc}.} The scale, $\rho^{(l)}$, is therefore important: a scale that is too small or too big forces the hidden-layer activities into a range that can hinder learning, as seen in Fig.~\ref{fig:app:constraint_scale}a. Interestingly, the BioCCPC algorithm with a well-chosen constraint scale converges faster than a plain PC network, that is, it can reach a lower PC loss given the same number of training samples (Fig.~\ref{fig:app:constraint_scale}a).

The predictive-coding loss, Eq.~\eqref{eq:obj_doubled}, penalizes prediction errors in all layers. If we think about the task as supervised learning, however, it might make more sense to focus only on the accuracy of the prediction in the last layer. The best output prediction error may require a different constraint scale compared to the lowest PC loss, cf.~Fig.~\ref{fig:app:constraint_scale}b and Fig.~\ref{fig:app:constraint_scale}d.

\subsection{Effect of dataset choice}
\begin{figure}[!tbh]
    \center\includegraphics{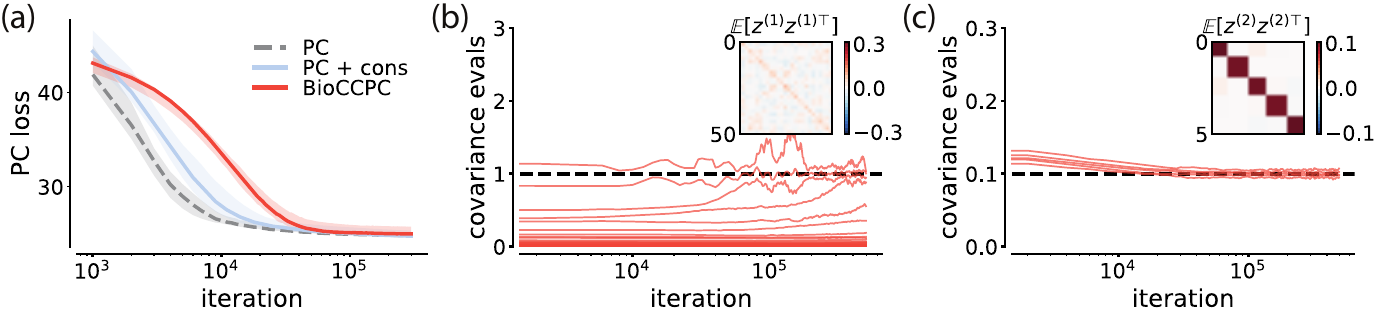}
    \caption{Comparison of our algorithm with PC on the Mediamill dataset. Simulations run using a network with two hidden layers, sizes 50 and 5. Shaded area represents the variation over 10 runs. (a) Evolution of the predictive-coding loss, Eq.~\eqref{eq:obj_doubled}, on a validation set during learning.
    (b) and (c) Evolution of the eigenvalues of the covariance matrix of hidden-layer activations in the two hidden layers. The dashed lines show the scale of the constraint; see SM Sec.~\ref{app:upperbound}. Inset: covariance matrix at end of training. Note how activity is whitened in the second hidden layer (c), and low-rank in the first hidden layer (b). Note also that the network performance is good even though the constraint is not yet fully converged.\label{fig:app:mmill}}
\end{figure}

While we derived our algorithm by starting from the supervised PC framework of \cite{Whittington2017}, the same derivation would hold for the unsupervised or self-supervised PC networks. In these settings, PC can be thought of as a framework which combines information from the inputs at the first and last layer, rather than a method for predicting an output from an input. In the linear case, this has much in common with biologically plausible circuits that perform canonical correlation \mbox{analysis}~(CCA)~\cite{RRR,BioCCA}. We therefore tested the algorithm on Mediamill~\citeSM{snoek2006challenge}, a dataset consisting of two views of the same scene, one given by the video and the other by text annotation. The results (Figs.~\ref{fig:app:constraint_scale}b-d and Fig.~\ref{fig:app:mmill}) show similar trends as those obtained on MNIST. In particular, an appropriate choice of constraint scale leads to results that are competitive with less biologically plausible implementations of PC.

In Fig.~\ref{fig:app:extra} we show the evolution of the validation-set loss for four other datasets: FashionMNIST, CIFAR10, CIFAR100, and the LFW (Labeled Faces in the Wild).  The first three are supervised-learning scenarios and we use a network similar to the one we employed for MNIST. The LFW dataset contains pairs of views of the same subject, and so we employ a symmetric autoencoder-like architecture for this self-supervised task. As before, we see that the BioCCPC algorithm converges to a solution that has predictive-coding loss comparable to the minimum obtained from traditional PC implementations (see Fig.~\ref{fig:app:extra}).

\begin{figure}[!tbh]
    \center\includegraphics{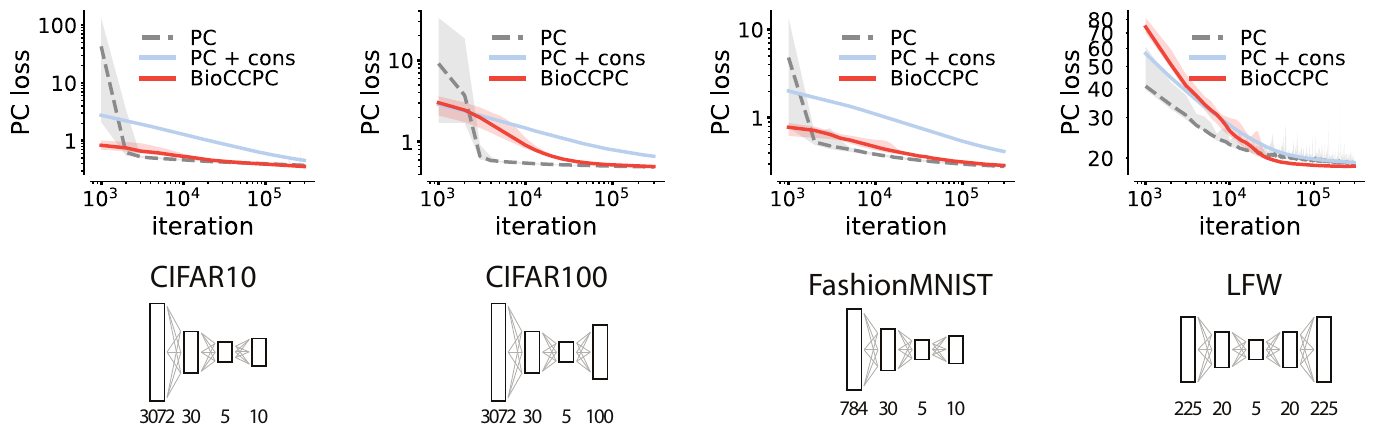}
    \caption{Comparison of our algorithm with PC on a variety of datasets. The name of the dataset and the architecture we employed are shown below each figure. Shaded area represents the variation over 10 runs. \label{fig:app:extra}}
\end{figure}

\subsection{Network architecture}
We tested BioCCPC across different architectures, such as a shallow network with a small hidden-layer dimension (Fig.~\ref{fig:app:constraint_scale}) and a deeper network with two hidden layers, one of which is large-dimensional (Fig.~\ref{fig:app:constraint_scale}a,b,d). In both cases BioCCPC compares well to PC.

\subsection{Lack of symmetry in feedforward and feedback weights}
One of the difficulties with a biological implementation of predictive coding is that it requires symmetric weights across layers. This is not so much an issue of plausibility---indeed, symmetric weights can be obtained simply by using the same Hebbian rules to train both feedforward and feedback connections---but rather a problem with realism: weights in many areas of the brain are not in fact seen to be symmetric.
\begin{figure}[!tbh]
    \center\includegraphics{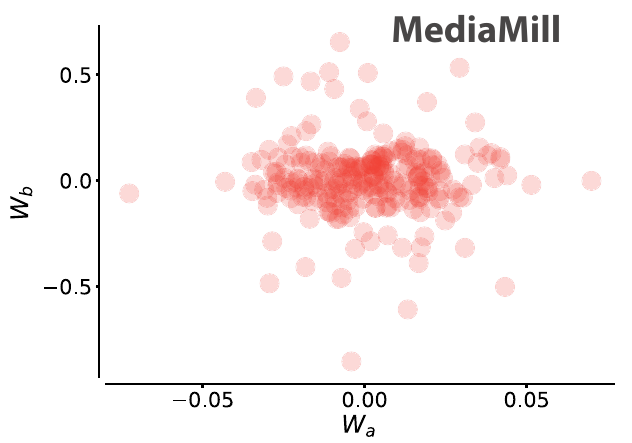}
    \caption{Comparison of feedforward and feedback weights after learning on a network with hidden layers of sizes 50 and 5 trained on the Mediamill dataset. \label{fig:app:weight_asymmetry_mmill}}
\end{figure}

\begin{figure}[!tbh]
    \center\includegraphics{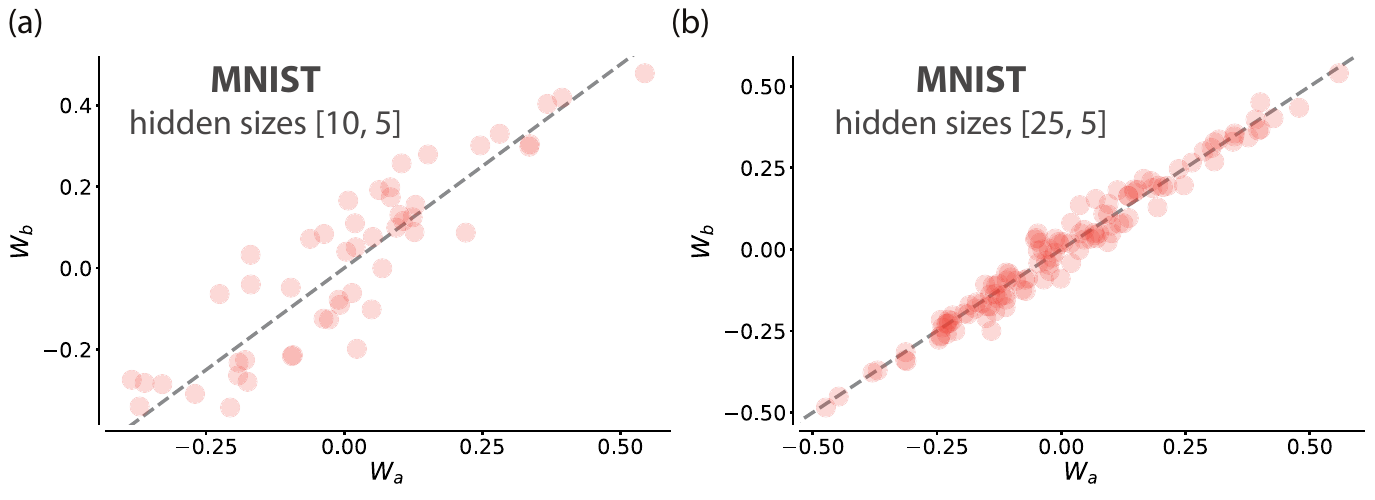}
    \caption{Comparison of feedforward and feedback weights after learning for different network sizes. This is trained on MNIST. (a) There is less symmetry in a smaller network with hidden layers of sizes 10 and 5. (b) There is more symmetry in a larger network, with hidden layers 25 and 5. \label{fig:app:weight_asymmetry}}
\end{figure}

Our algorithm explicitly breaks the symmetry between the weights, leading to solutions that generically  alleviate this difficulty. Indeed, in some cases there is very little correlation between feedforward and feedback weights; see Fig.~\ref{fig:app:weight_asymmetry_mmill}. However, in some cases we see that the forward and backward weights converge to similar values even though the learning rules are not symmetric. We suspect that this can occur in cases where there is a large imbalance between the width of adjacent layers, the larger the imbalance (i.e. the lower rank the wider layer is) the more correlation we see in our experiments; see Fig.~\ref{fig:app:weight_asymmetry}. In biologically relevant neural networks, we do not generally see this large imbalance in the width of adjacent layers and therefore we do not expect to see high correlation between the forward and backward weights.

\subsection{Predictive coding with constraint}
In the main text we considered a predictive-coding network that imposes a covariance constrained but we did not make explicit what this entails. Explicitly, we augmented the Whittington-Bogacz objective, Eq.~\eqref{eq:WB_obj}, with a Lagrange multiplier enforcing an inequality constraint on the covariance of hidden variables, similar to what we do in BioCCPC:
\begin{align}\label{eq:WB_obj_with_cons}
    L_\text{cons} &= \frac12 \sum_{l} \frac{\bigl\| \Z^{(l)}-\W^{(l-1)} \Z^{(l-1)}\bigr\|_F^2}{\sigma^{(l)2}} + \tr \Q^{(l)\top} \Q^{(l)} (\Z^{(l)} \Z^{(l)\top} - T \rho^{(l)} \I),
\end{align}
and optimized using stochastic gradient descent. This lets us distinguish the effect of the constraint from the effect of the approximations we used in the derivation.

\subsection{Details of simulation runs}
We initialized our the weights in our networks using uniform random numbers in the range $[-a, a]$, where $a = 1 / \sqrt{\text{number of columns}}$. This is similar to Xavier initialization~\citeSM{glorot2010understanding} except the scale of initial weights depends only on the number of columns instead of both rows and columns. We have found this initialization to provide slightly better results.

We optimized hyperparameters using Optuna~\citeSM{optuna_2019} with the tree-structured Parzen Estimator (TPE) algorithm~\citeSM{bergstra2011algorithms}. We ran the optimization for 100 trials and repeated it 10 times with different random seeds for every combination of dataset, network architecture, and constraint scale $\rho$ (with the exception of the plain PC algorithm which does not have a constraint). In each case we optimized the learning rate $\eta$ used in the weight updates and the learning rate $\tau^{-1}$ of the fast dynamics (see Algorithm~\ref{alg:constraint_PC}). For networks that have a constraint, we allowed a different learning rate $f \eta$ for the $\Q$ dynamics, and optimized for the factor $f$. Finally, for BioCCPC network, we allowed for different learning rates for the forward and backward weights, $\W_b$ and $\W_a$, and optimized for the ratio between these learning rates. Hyperparameter optimization runs were trained for 500 batches with batch size 100 and were evaluated on a held-out validation set comprising 500 samples. Each run was repeated 4 times with different random initialization to account for stochasticity in the learning dynamics.

The simulation runs that were used to make the figures in the paper were based on the hyperparameter optimization results that yielded the lowest predictive-coding loss. Since hyperparameter optimization tends to push the learning rate up to the brink of instability, the learning rate used in the simulations was lowered by multiplying by a certain ``safety'' factor $\alpha$. We generally used $\alpha = 0.8$.

We ran our simulations on the CPU. Each run took from 1 to 4 minutes when run single-threaded on a 24-core Intel Xeon CPU E5-2643 at 3.40GHz running Linux. We also used a cluster to run the hyperparameter optimization for different architectures in parallel.

\section{Generalization error}
\begin{figure}[!tbh]
    \center\includegraphics{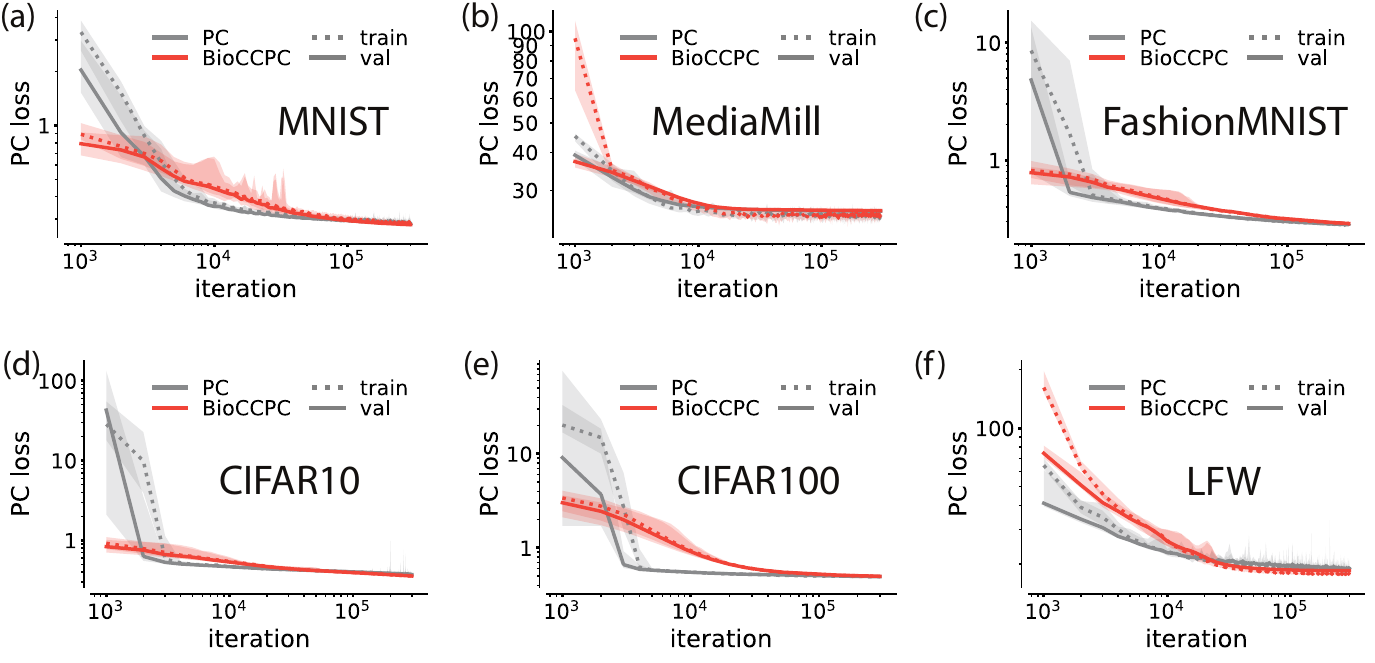}
    \caption{There is little overfitting in the linear setting that we are investigating. The plots compare the training- and validation-set predictive-coding losses for both the less biologically plausible predictive-coding implementation from~\cite{Whittington2017} (PC) and for our algorithm (BioCCPC).\label{fig:app:generalization}}
\end{figure}
The linear constraint that we impose on our network is a strong regularizer. Figure~\ref{fig:app:generalization} compares the training and validation losses during training on six different datasets (as described above). We see that there is very little overfitting in all cases.

\section{Addition of a} \label{app:addition}

:
\begin{align*}
     \min_{\Z^{(l)},\W_a^{(l)},\W_b^{(l)}} \hat L^{(l)}  \leq&  \min_{\substack{\Z^{(l)},\W_a^{(l)},\W_b^{(l)} \\ \Z^{(l)} \Z^{(l)\top} \preceq T\times\I}}\hat L^{(l)} \\ \leq& \min_{\substack{\Z^{(l)},\W_a^{(l)},\W_b^{(l)} \\ \Z^{(l)} \Z^{(l)\top\preceq T\times\I}}}  \frac 12\biggl[g_b^{(l)} \Bigl\|\Z^{(l)}-\W_b^{(l-1)}\Z^{(l-1)}\Bigr\|_F^2 
    + c^{(l)} \bigl\|\Z^{(l)}\bigr\|_F^2 \\
    &\qquad+ g_a^{(l)}   \tr \Bigl(- 2 \Z^{(l+1)\top}\W_a^{(l)}\Z^{(l)}\Bigr)  + g_a^{(l)}  T \, \W_a^{(l)\top} \W_a^{(l)}
\end{align*}
The first inequality results from the fact that the parameters of the new optimization problem are a subset of the parameters of the original optimization problem. In the second inequality we have expanded the second term of $\hat L^{(l)}$ and used the constraint $\frac 1T \Z^{(l)} \Z^{(l)\top} \preceq \I$ to replace $\Z^{(l)} \Z^{(l)\top}$ with~$T$, making the term larger, resulting again in an upper bound.

\bibliographystyleSM{unsrt}
\bibliographySM{biblio}

\end{document}